\documentclass{sig-alternate}

\usepackage{amssymb,amsfonts,amsmath,xspace}
\usepackage[naturalnames]{hyperref} 
\usepackage{algorithmic} 
\usepackage[english, algosection, algoruled, noline]{algorithm2e}
\newtheorem{theorem}{Theorem}[section]
\newtheorem{lemma}[theorem]{Lemma}
\newtheorem{proposition}[theorem]{Proposition}
\newtheorem{corollary}[theorem]{Corollary} 
\newtheorem{definition}[theorem]{Definition}

\def\newblock{\hskip .11em plus .33em minus .07em}
\newenvironment{bc}{\noindent{\bf Bivariate case:}}{\hfill$\Box$\medskip}
	{\hfill$\Box$\smallskip\end{trivlist}\vspace*{-.6cm}}

\def\bm{\boldsymbol}
\def\axel{{\sc Axel}}
\def\mathemagix{{\sc Mathemagix}}
\newcommand{\ds}{\displaystyle}
\newcommand{\dott}{{..}}
\newcommand{\hide}[1]{{}}

\def\CC{{\mathbb C}}  \def\PP{{\mathbb P}}
\def\QQ{{\mathbb Q}} \def\RR{{\mathbb R}} \def\ZZ{{\mathbb Z}}

\def\QQ{\mathbb{Q}}
\def\kk{\mathbb{K}}
\def\RR{\mathbb{R}}

\def\sign{\mathrm{sign}}
\def\flower{\mathrm{low}}
\def\fupper{\mathrm{upp}}

\newcommand{\OO}{\ensuremath{\mathcal{O}}\xspace}
\newcommand{\sO}{\ensuremath{\widetilde{\mathcal{O}}}\xspace}
\newcommand{\OB}{\ensuremath{\mathcal{O}_B}\xspace}
\newcommand{\sOB}{\ensuremath{\widetilde{\mathcal{O}}_B}\xspace}

\newcommand{\uvec}[1]{\underline{#1}}

\newcommand{\dg}{\ensuremath{\mathsf{deg}}\xspace}
\newcommand{\bitsize}[1]{\ensuremath{\mathcal{L}\left( #1 \right)}\xspace}

\title{
Continued Fraction Expansion of Real Roots of Polynomial Systems
}

\numberofauthors{1}

\author{
\alignauthor
Angelos Mantzaflaris, Bernard Mourrain, Elias Tsigaridas\\
\affaddr{GALAAD, INRIA Sophia Antipolis}\\
\email{[FirstName.LastName]@sophia.inria.fr}
}

\begin{document}

\conferenceinfo{SNC'09,} {August 3--5, 2009, Kyoto, Japan.} 
\CopyrightYear{2009}
\crdata{978-1-60558-664-9/09/08} 

\maketitle

\begin{abstract}
  We present a new algorithm for isolating the real roots of a system of
  multivariate polynomials, given in the monomial basis. It is inspired by
  existing subdivision methods in the Bernstein basis; it can be seen as
  generalization of the univariate continued fraction algorithm or
  alternatively as a fully analog of Bernstein subdivision in the monomial basis.
  The representation of the subdivided domains is done through homographies,
  which allows us to use only integer arithmetic and to treat efficiently
  unbounded regions.  We use univariate bounding functions, projection and
  preconditionning techniques to reduce the domain of search. The resulting
  boxes have optimized rational coordinates, corresponding to the first terms
  of the continued fraction expansion of the real roots.  An extension of
  Vincent's theorem to multivariate polynomials is proved and used for the
  termination of the algorithm. New complexity bounds are provided for a
  simplified version of the algorithm. Examples computed with 
  a preliminary C++ implementation illustrate the approach.
\end{abstract}

\vspace{-.3cm}
\vspace{-.3cm}
\vspace{-.3cm}


\section{Introduction} \label{intro}

The problem of computing roots of univariate polynomials has a long
mathematical history \cite{Pan97rev}.  Recently, some new
investigations focused on subdivision methods, where root localization
is based on simple tests such as \emph{Descartes' Rule of Signs} and
its variant in the Bernstein basis \cite{MOURRAIN:2005, emt-lncs-2006,
  MR2289103}.  Complexity analysis was developed for univariate
integer polynomial taking into account the bitsize of the
coefficients, and providing a good understanding of their behavior
from a theoretical and practical point of view.  Approximation and
bounding techniques have been developed \cite{jue:07e} to improve the
local speed of convergence to the roots.

Even more recently a new attention has been given to continued
fraction algorithms (CF), see
e.g. \cite{sharma-tcs-2008,et-tcs-2007} and references
therein. They differ from previous subdivision-based algorithms in
that instead of bisecting a given initial interval and thus producing
a binary expansion of the real roots, they compute continued fraction
expansions of these roots. The algorithm relies heavily on
computations of lower bounds of the positive real roots, and different
ways of computing such bounds lead to different variants of the
algorithm.  The best known worst-case complexity of CF is $\sOB( d^5
\tau^2)$ \cite{sharma-tcs-2008}, while its average complexity is
$\sOB( d^3 \tau)$, thus being the only complexity result that matches,
even in the average the complexity bounds of numerical algorithms
\cite{Pan02jsc}.  Moreover, the algorithm seems to be the most
efficient in practice \cite{ACS-TR-363602-02,et-tcs-2007}.

Subdivision methods for the approximation of isolated roots of
multivariate systems are also investigated but their analysis is much
less advanced.  In \cite{sp-csnps-93}, the authors used tensor product
representation in Bernstein basis and domain reduction techniques
based on the convex hull property to speed up the convergence and
reduce the number of subdivisions. In \cite{ELBER:2001}, the emphasis
is put on the subdivision process, and stopping criterion based on the
normal cone to the surface patch.  In \cite{mp:smspe-05}, this
approach has been improved by introducing pre-conditioning and
univariate-solver steps. The complexity of the method is also analyzed
in terms of intrinsic differential invariants.

This work is in the spirit of \cite{mp:smspe-05}.  The novelty of our
approach is the presentation of a tensor-monomial basis algorithm that
generalizes the univariate continued fraction algorithm and does not
assume generic position. We apply a subdivision approach also
exploiting certain properties of the Bernstein polynomial
representation, even though no basis conversion takes place.


Our contributions are as follows.
We propose a new adaptive algorithm for polynomial system real solving
that acts in monomial basis, and exploits the continued fraction
expansion of (the coordinates of) the real roots.  This yields the
best rational approximation of the real roots.  All computations are
performed with integers, thus this is a division-free algorithm.
We propose a first step towards the generalization of Vincent's
theorem to the multivariate case (Th.~\ref{vincentxyz})
We perform a (bit) complexity analysis of the algorithm, when oracles
for lower bounds and counting the real roots are available
(Prop.~\ref{prop:mcf-complexity}) and we propose non-trivial
improvements for reducing the total complexity even more
(Sec.~\ref{sec:complexity-improvements}).  In all cases the bounds
that we derive for the multivariate case, match the best known ones
for the univariate case, if we restrict ourselves to $n=1$.

\subsection{Notation}

For a polynomial $f \in$ $\RR[x_1, \dott, x_n]$, $\dg( f)$ denotes its
total degree, while $\dg_{x_i}(f)$ denotes its degree w.r.t.~$x_i$.
Let $f(\uvec x)=f(x_1,\dott, x_n )\in \RR[x_1,\dott, x_n]$ with $\dg_{x_k} f =
d_k$, $k=1,\dott,n$. 
If not specified, we denote $d=d(f)=\max\{d_1,\dott,d_n\}$.

We are interested in isolating the real roots of a system of polynomials 
$f_{1}(\uvec x), \dott, f_{s}(\uvec x)\in \ZZ[x_{1}, \dott, x_{n}]$, in a
box $I_{0}= [u_1,v_1]\times\cdots \times [u_n, v_n] \subset \RR^n$, $u_k,v_k\in\QQ$. 
We denote by $\mathcal{Z}_{\kk^{n}}(f)=\{ p\in \kk^{n}; f(p)=0\}$ the solution
set in $\kk^{n}$  of the equation $f(x)=0$, where $\kk$ is $\RR$ or $\CC$.

In what follows \OB, resp. \OO, means bit, resp. arithmetic,
complexity and the \sOB, resp. \sO, notation means that we are
ignoring logarithmic factors.
For $a \in \QQ$, $\bitsize{ a} \ge 1$ is the maximum bit size of the
numerator and the denominator.  For a polynomial $f \in$ $\ZZ[x_1,
\dott, x_n]$, we denote by \bitsize{f} the maximum of the bitsize of
its coefficients (including a bit for the sign).  In the following, we
will consider classes of polynomials such that $\log( d(f) ) = \OO(
\bitsize{f})$.


Also, to simplify the notation we introduce multi-indices, for the
variable vector $\uvec x=(x_1,\dott,x_n)$, $\uvec x^{\uvec i}
:= x_1^{i_1}\cdots x_n^{i_n}$, the sum $\ds \sum_{\uvec i= \uvec 0}^{\uvec
  d}:= \sum_{i_1=0}^{d_1}\cdots\sum_{i_n=0}^{d_n}$, and $\ds \binom{\uvec
  d}{\uvec i} := \binom{d_1}{i_1}\cdots \binom{d_n}{i_n} $.
The tensor Bernstein basis polynomials of multidegree degree $\uvec d$
of a box $I$ are denoted $\ds B(\uvec x; \uvec i, \uvec d; I) :=
B_{d_1}^{i_1}(x_1; u_1,u_1)\cdots$ $B_{d_n}^{i_n}(x_n;u_n,u_n)$ where
$I= [\uvec u,\uvec v]:=[u_1,v_1]\times\cdots\times[u_n,v_n]$.

\subsection{The general scheme}   

In this section, we describe the family of algorithms that we consider.
The main ingredients are 
\begin{itemize}
 \item  a suitable representation of the equations in a given (usually
rectangular) domain, for instance a representation in the Bernstein
basis or in the monomial basis;
 \item an algorithm to split the
representation into smaller sub-domains;
 \item a reduction procedure to
shrink the domain.
\end{itemize}
Different choices for each of these ingredients lead to algorithms
with different practical behaviors. The general process is summarized
in Alg.~\ref{algo:subdivision}.

\begin{algorithm} \label{algo:subdivision}
\caption{Subdivision scheme}
\KwIn{A set of equations $f_1, f_2, \dott, f_s\in \ZZ[\uvec x]$ represented over a domain $I$.} 
\KwOut{A list of disjoint domains, each containing one and only one real root of $f_1=\cdots=f_s=0$.}

Initialize a stack $Q$ and add $(I,f_1,\dott,f_s)$ on top of it\;

While $Q$ is not empty do 
\begin{enumerate}
\item[a)] Pop a system $(I,f_1,\dott,f_s)$ and:
\item[b)] Perform a precondition process and/or a reduction process to
  refine the domain.
\item[c)] Apply an exclusion test to identify if the domain contains
  no roots.
\item[d)] Apply an inclusion test to identify if the domain contains a
  single root. In this case output $(I,f_1,\dott,f_s)$.
\item[e)] If both tests fail split the representation into a number of
  sub-domains and push them to $Q$.
\end{enumerate}
\end{algorithm}
The instance of this general scheme that we obtain generalizes the
continued fraction method for univariate polynomials; the realization of the
main steps (b-e) can be summarized as follows: 
\begin{enumerate}
\item[b)] Perform a precondition process and compute a lower bound on
  the roots of the system, in order to reduce the domain.
\item[c)] Apply interval analysis or sign inspection to identify if some
  $f_i$ has constant sign in the domain, i.e. if the domain contains
  no roots.
\item[d)] Apply Miranda test to identify if the domain contains a
  single root. In this case output $(I,f_1,\dott,f_s)$.
\item[e)] If both tests fail, split the representation at $(1,\dott,1)$
  and continue.
\end{enumerate}

In the following sections, we are going to describe more precisely the
specific steps and analyze their complexity.
In Sec.~\ref{homography}, we describe the representation of domains via
homographies and the connection with the Bernstein basis representation.
Subdivision, based on shifts of univariate polynomials, reduction and
preconditionning are analyzed in Sec.~\ref{subdivision-reduction}.
Exclusion and inclusion tests as well as a generalization of Vincent's
theorem to multivariate polynomials, are presented in
Sec.~\ref{criteria}.
In Sec. \ref{sec:complexity}, we recall the main properties of
Continued Fraction expansion of real numbers and use them to analyze
the complexity of a subdivision algorithm following this generic
scheme.
We conclude with examples produced by our C++ implementation in
Sect.~\ref{sec:impl}.

\section{Representation: Homographies} \label{homography}

A widely used representation of a polynomial $f$ over a rectangular
domain is the tensor-Bernstein representation. De Casteljau's
algorithm provides an efficient way to split this representation to
smaller domains. A disadvantage is that converting integer polynomials
to Bernstein form results in rational or, if one uses machine numbers,
approximate Bernstein coefficients. We follow an alternative approach
that does not require basis conversion since it applies to monomial
basis: We introduce a tensor-monomial representation, i.e. a
representation in the monomial basis over $\PP^1\times\cdots\times
\PP^1$ and provide an algorithm to subdivide this representation
analogously to the Bernstein case.

In a tensor-monomial representation a polynomial is represented as a
tensor (higher dimensional matrix) of coefficients in the natural
monomial basis, that is,
$$
f(\uvec x)= \sum_{i_1,\dott,i_n}^{d_1,\dott,d_n} c_{i_1\dott i_n} \uvec x^{(i_1,\dott,i_n)} = \sum_{\uvec i=\uvec 0}^{\uvec d} c_{\uvec i} \uvec x^{\uvec i} ,
$$
for every equation $f$ of the system.  Splitting this
representation is done using homographies. The main operation in this
computation is the Taylor shift.

\begin{definition}
  A homography (or Mobius transformation) is a bijective projective
  transformation $\mathcal H = (\mathcal H_1,\dott,\mathcal H_n)$,
  defined over $\PP^1\times\cdots\times \PP^1$ as
$$ 
x_k \mapsto \mathcal H_k(x_k)= \frac{\alpha_k x_k + \beta_k}{ \gamma_k x_k + \delta_k } 
$$ 
with $\alpha_k,\beta_k,\gamma_k,\delta_k \in \ZZ$, $\gamma_k \delta_k\ne 0 $, $k=1,\dott,n$. 
\end{definition}
Using simple calculations, we can see that the inverse
$$\ds \mathcal H_k^{-1}(x_k)= \frac{-\delta_k x_k + \beta_k}{ \gamma_k x_k - \alpha_k } 
$$
is also a homography, hence the set of homographies is a group under
composition.  Also, notice that if $\det \mathcal H>0$ then, taking
proper limits when needed, we can write
\begin{equation}\label{inverse_hg}
\RR_+  \mapsto \mathcal H_k( \RR_+ )=  \left[ \frac{\beta_k}{\delta_k},  \frac{\alpha_k}{\gamma_k}\right]
\end{equation}
hence $ H(f)\ : \ \RR^n_+ \to \RR $, 
$$
H(f) :=\prod_{k=1}^n(\gamma_kx_k
+ \delta_k)^{d_k} \cdot (f\circ \mathcal H) (x) 
$$
is a polynomial defined
over $\RR^n_{+}$ that corresponds to the (possibly unbounded) box 
\begin{equation}\label{Hbox}
I_H = \mathcal H( \RR_+^n) = 
\left[ \frac{\beta_1}{\delta_1},  \frac{\alpha_1}{\gamma_1}\right]
\times \cdots\times 
\left[ \frac{\beta_n}{\delta_n},  \frac{\alpha_n}{\gamma_n}\right]  ,
\end{equation}
of the initial system, in the sense that the zeros of the initial
system in $I_H$ are in one-to-one correspondence with the positive
zeros of $H(f)$.

We focus on the computation of $H(f)$. 
We use the basic homographies $T_{k}^{c}(f)= f|_{x_k=x_k+c} $
(translation by $c$) or simply $T_k(f)$ if $c=1$, $C_k^c( f ) =
f|_{x_k=cx_{k}} $ (contraction by $c$) and $R_{k}(f)=
x_k^{d_k}f|_{x_k=1/x_{k}}$ (reciprocal polynomial).
These notations are naturally extended to variable vectors; for instance
$T^{\uvec c}=(T_1^{c_1},\dott,T_n^{c_n} )$, $c=(c_1,\dott,c_n)\in\ZZ^n$.
Complexity results for these computations appear in the following sections.
We can see that they suffice to compute any homography:

\begin{lemma}
The group of homographies with coefficients in $\ZZ$ is generated by $R_k,C_k^c,T_k^c$, $k=1,\dott,n$, $c\in\ZZ$.
\end{lemma}
\begin{proof}
It can be verified that a $H_k(f)$ with arbitrary coefficients $\alpha_k, \beta_k, \gamma_k, \delta_k \in \ZZ $ is constructed as
$$
H_k(f) =  C_k^{\gamma_k}R_kC_k^{\delta_k}T_kR_kC_k^{ \alpha_k/\gamma_k -  \beta_k/\delta_k } T_k^{\beta_k/\delta_k} (f)
$$
where the product denotes composition.  We abbreviate 
$C_k^{1/c}=R_kC_k^cR_k$ 
and 
$T_k^{u/c}= C_k^cT_k^uC_k^{1/c}$, $u,c\in\ZZ$,
e.g. $C_k^{1/c}(x)=\frac x c$ and $T_k^{u/c}(x)=x+\frac u c$.

\end{proof}

Representation via homography is in an interesting correspondence to
the Bernstein representation:

\begin{lemma} \label{lem:bernsteincoefs}
Let $f=\sum_{\uvec i=0}^{\uvec d} b_i\, B_i^n(\uvec x,I_H)$ the Bernstein expansion of $f$ in the box $I_H$ yielded by a homography $H$. If 
$$
H(f) = C^{\gamma}RC^{\delta}T^1RC^{\alpha/\gamma - \beta/\delta} T^{\beta/\delta} ( f ) = \sum_{\uvec i=0}^{\uvec d} c_i \uvec x^{\uvec i}
$$
then
$\ds
 c_i = \binom{\uvec d}{\uvec i}\uvec \gamma^{\uvec i} \uvec \delta^{\uvec d-\uvec i}  b_i  .
$
\end{lemma}
\begin{proof}
Let $[u_k,v_k]=\left[ \frac{\beta_k}{\delta_k} , \frac{\alpha_k}{\gamma_k} \right] $. For a tensor-Bernstein polynomial $\ds \binom{\uvec d}{\uvec i} \frac{1}{(v-u)^d}(x-u)^i(v-x)^{d-i}$ we compute
\begin{align*}
                       &   C^{\gamma}RC^{\delta}T^1RC^{v-u} T^{u}    (  \binom{\uvec d}{\uvec i} \frac{1}{(v-u)^d}(x-u)^i(v-x)^{d-i} )  \\ 
                       &=  C^{\gamma}RC^{\delta}RT^{1}RC^{v-u}       (  \binom{\uvec d}{\uvec i} \frac{1}{(v-u)^d}x^i(v-u-x)^{d-i} )    \\
                       &=  C^{\gamma}RC^{\delta}RT^{1}              (  \binom{\uvec d}{\uvec i} (x-1)^{d-i})                           \\
                       &=  C^{\gamma}RC^{\delta}                    (  \binom{\uvec d}{\uvec i}  x^{i})                                
                       =                                            \binom{\uvec d}{\uvec i}\gamma^i \delta^{d-i} x^{i} 
\end{align*}
as needed.
\end{proof}

\begin{corollary}\label{cor:bernsteincoefs}
The Bernstein expansion of $f$ in $I_H$ is
$$
\sum_{i=0}^{d} \frac{ c_{i}}{\binom{\uvec d}{\uvec i} \uvec \gamma^{\uvec i} \uvec \delta^{\uvec d-\uvec i} } B (\uvec x ; {\uvec i},{\uvec d}; I_H)    .
$$
That is, the coefficients of $H(f)$ coincide with the Bernstein
coefficients up to contraction and binomial factors.
\end{corollary}

Thus tensor-Bernstein coefficients and tensor-mo\-no\-mial
coefficients in a sub-domain of $\RR_+^n$ differ only by
multiplication by positive constant. In particular they are of the
same sign. Hence this corollary allows us to take advantage of sign
properties (eg. the variation diminishing property) of the Bernstein
basis without computing it.


The resulting representation of the system consists of  the transformed
polynomials $H(f_1),\dott,H(f_n)$, represented as tensors of coefficients as
well as $4n$ integers, $\alpha_k,\beta_k,\gamma_k,\delta_k$ for $k=1,\dott,n$
from which we can recover the endpoints of the domain, using~(\ref{Hbox}).

\section{Subdivision and reduction} \label{subdivision-reduction}

\subsection{The subdivision step}

We describe the subdivision step using the homography representation. This is
done at a point $\uvec u=(u_1,\dott,u_n)\in \ZZ^n_{\ge 0}$. It
consists in computing up to $2^n$ new sub-domains (depending on the
number of nonzero $u_k$'s), each one having $\uvec u$ as a vertex.

Given $H(f_1),\dott,H(f_s)$ that represent the initial system at some
domain, we consider the partition of $\RR^n_+$ defined by the
hyperplanes $x_k=u_k$, $k=1,\dott,n$. These intersect at $\uvec u$
hence we call this \emph{partition at $\uvec
  u$}. Subdividing at $\uvec u$ is equivalent to subdividing the
initial domain into boxes that share the common vertex $\mathcal
H(\uvec u)$ and have faces either parallel or perpendicular to those of
the initial domain.

We need to compute a homography representation for every domain in
this partition. The computation is done coordinate wise; observe that
for any domain in this partition we have, for all $k$, either
$x_k\in[0,u_k]$ or $x_k\in[u_k,\infty]$. It suffices to apply a
transformation that takes these domains to $\RR_+$. In the former case,
we apply $T_k^{1}R_kC_k^{u_k}$ to the current polynomials and in the latter
case we shift them by $u_k$, i.e. we apply $T_k^{u_k}$. The integers
$\alpha_k,\beta_k,\gamma_k,\delta_k$ that keep track of the current
domain can be easily updated to correspond to the new subdomain.

We can make this process explicit in general dimension: every computed subdomain
corresponds to a binary number of length $n$, where the $k-$th bit is
$1$ if $T_k^{1}R_kC^{u_k}$ is applied or $0$ if $T_k^{u_k}$ is applied. 

In our continued fraction algorithm the subdivision is performed at $\uvec u=\uvec 1$.

\noindent{}{\bf Illustration.} Let us illustrate this process in dimension
two. The system $f_1, f_2$ is defined over $\RR^2_{>0}$. We subdivide this
domain into $[0,1]^2$, $[0,1]\times \RR_{>1}$, $\RR_{>1}\times [0,1]$ and
$\RR_{>1}\times\RR_{>1}$. Equivalently, we compute four new pairs of
polynomials, as illustrated in Fig.~\ref{fig:subdiv} (we abbreviate
$S_k=T_k^{1}R_k$).
\begin{figure}[h]
  \centering
  \includegraphics[scale=.6]{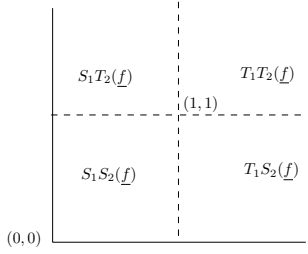}
  \caption{Subdividing the domain of $\uvec f$.}
  \label{fig:subdiv}
\end{figure}


\noindent{}{\bf Complexity of subdivision step.}
The transformation of a polynomial into two sub-domains,
i.e. splitting w.r.t. one direction, consists of performing $d^{n-1}$
univariate shifts, one for every coefficient $\in \ZZ[x_k]$ of
$f \in \ZZ[x_k][x_1,\dott,\widehat{x_k},\dott, x_n] $.

If the subdivision is performed in every direction, each
transformation consists of $d^{n-1}$ univariate shifts for every
variable, i.e. $nd^{n-1}$ shifts. There are $2^n$ sub-domains to
compute, hence a total of $n^2 2^n d^{n-1} $ shifts have to be
performed in a single subdivision step. We must also take into account
that every time a univariate shift is performed, the coefficient
bitsize increases.

The operations 
$$
T_k(f)= f|_{x_k=x_k+1} \ \text{ and } \ 
T_kR_k(f)=(x_k+1)^{d_k} f|_{x_k=\frac{1}{x_k+1}}
$$
are essentially of the same complexity, except that the second
requires one to exchange the coefficient of $c_{i_1,\dott,i_k,\dott,i_n}$
with $c_{i_1,\dott,d_k-i_k,\dott,i_n}$ before translation, i.e. an
additional $\OO(d^n)$ cost. Hence we only need to consider the case of
shifts for the complexity.

The continued fraction algorithm subdivides a domain using unit shifts
and inversion.  Successive operations of this kind increase the bitsize
equivalently to a big shift by the sum of these units. Thus it
suffices to consider the general computation of $f(\uvec x + \uvec u)$
to estimate the complexity of the subdivision step.

\begin{lemma}[{Shift complexity}] \label{shiftComplexity} The
  computation of $f(\uvec x + \uvec u)$ with $\mathcal L(f)=\tau$ and
  $\mathcal L(u_k)\leq\sigma,\ k=1,\dott,n$ can be performed in $\sOB(
  n^2 d^n \tau + d^{n+1} n^3 \sigma )$.
\end{lemma}
\begin{proof}
  We use known facts for the computation of $T_k^{u_k}(f)$ for
  univariate polynomials.  If $\text{deg}_k f = d_k$ and $f$ is
  univariate, this operation is performed in
  $\sOB(d_k^2\sigma+d_k\tau)$; the resulting coefficients are of
  bitsize $ \tau +d_k\sigma $ \cite{gg-shift-1997}. 
  Hence $f(x_1,\dott,x_k+u_k,\dott,x_n)$ is
  computed in $\sOB(d^{n-1}(d_k^2\sigma+d_k\tau) )$.

  Suppose we have computed $f(x_1+u_1,x_{k-1}+u_{k-1},x_k,\dott,x_n)$
  for some $k$. The coefficients are of bitsize $\tau
  +\sum_{i=1}^{k-1}\sigma_i$. The computation of shift w.r.t. $k-$th
  variable $f(x_1+u_1,\dott,x_k+u_k,x_{k+1},\dott,x_n)$ results in a
  polynomial of bitsize $\tau+ \sum_{i=1}^k \sigma_i$ and
  consists of $ d^{n-1} \sOB( d^2 \sum_{i=1}^k \sigma_i + d \tau ) )$
  operations. That is, we perform $d^{n-1}$ univariate polynomial
  shifts, one for every coefficient of $f$ in
  $\ZZ[x_k][x_1,\dott,\widehat{x_k},\dott, x_n] $.

  This gives a total cost for computing $f(\uvec x + \uvec u)$ of
$$
d^{n-1} \sum_{k=1}^n \left( d^2 \sum_{i=1}^k \sigma_i + d \tau \right)
= nd^n\tau + d^{n+1} \sum_{k=1}^n (n+1-k)\sigma_k .
$$
The latter sum implies that it is faster to apply the shifts with
increasing order, starting with the smallest number $u_k$. Since
$\sigma_k=\OO(\sigma)$ for all $k$, and we must shift a system of
$\OO(n)$ polynomials we obtain the stated result.
\end{proof}

Let us present an alternative way to compute a sub-domain using
contraction, preferable when the bitsize of $\uvec u$ is big.  The
idea behind this is the fact that $T_k^c$ and $T^1_kC^c$ compute the same
sub-domain, in two different ways. 

\begin{lemma} \label{contractComplexity} If $f=\sum_{\uvec i=\uvec
    0}^{\uvec d} c_{\uvec i} \uvec x^{\uvec i}$, $\mathcal L(f)=\tau$,
  then the coefficients of $C^u(f)$, $\mathcal L(u_k)\leq\sigma,\
  k=1,\dott,n$, can be computed in $  \sOB(  d^n\tau + nd^{n+1}\sigma) $ .
\end{lemma}
\begin{proof}
  The operation, i.e. computing the new coefficients $c_{\uvec i}
  \uvec u^{\uvec i}$ can be done with $\sO(d^n)$ multiplications: Since
  $\uvec u^{(i_1,\dott,i_k,\dott,i_n)}= u_k \uvec u^{(i_1,\dott,i_k-1,\dott,i_n)}$, 
  if these powers are computed
  successively then every coefficient is computed using two
  multiplications. Moreover, it suffices to keep in memory the $n$
  powers $u^{(i_1,\dott,i_{k-1},i_k-1,i_{k+1},\dott,i_n)}$,
  $k=1,\dott,n$ in order to compute any $\uvec u^{\uvec i} c_{\uvec
    i}$. Geometrically this can be understood as a stencil of $n$
  points that sweeps the coefficient tensor and updates every element
  using one neighbor at each time. The bitsize of the multiplied
  numbers is $\OO(\tau+d\sigma)$ hence the result follows.
\end{proof}

Now if we consider a contraction followed by a shift by $1$
w.r.t. $x_k$ for $\OO(n)$ polynomials we obtain $ \sOB(n^2d^n\tau +
n^3d^{n+1} + nd^{n+1}\sigma) $ operations for the computation of the
domain. The disadvantage is that the resulting coefficients are of
bitsize $\OO(\tau+d\sigma)$ instead of $\OO(\tau+n\sigma)$ with the
use of shifts. Also note that this operation would compute a expansion
of the real root which differs from  continued fraction expansion.


\subsection{Reduction: Bounds on the range of {$f$}}  \label{subsec:reduction}

In this section we define univariate polynomials whose graph in
$\RR^{n+1}$ bounds the graph of $f$. For every direction $k$, we provide
two polynomials bounding the values of $f$ in $\RR^n$ from
below and above respectively.

Define
\begin{eqnarray}
  m_k (f ; x_k) & = & \sum_{i_k = 0}^{d_k}
  \min_{i_1,..,\widehat{i_k},.., i_n } c_{i_1 \dott i_n} \, x_k^{i_k} \\ 
  M_k (f ; x_k) & = & \sum_{i_k = 0}^{d_k}
  \max_{i_1,..,\widehat{i_k},.., i_n } c_{i_1 \dott i_n} \, x_k^{i_k}
\end{eqnarray}

\begin{lemma} \label{u_bounds}
  For any $ x \in \RR^n_{+}$, $n>1$ and any $k = 1, \dott, n$, we have
\begin{equation}\label{univbounds}
   m_k (f ; x_k) \le \frac{ f ( x)}{\ds \prod_{s\ne k}  \sum_{i_s=0}^{d_s} x_s^{i_s} } \le M_k (f ; x_k)  . 
\end{equation}
\end{lemma}
\begin{proof}
For $x \in \RR_{+}^{n}$, we can directly write
$$  f ( x) \leq \left( 
\sum_{i_k = 0}^{d_k} \max_{i_1,..,\widehat{i_k},.., i_n } c_{i_1 \dott i_n} \, x_k^{i_k}\right) 
\prod_{s\ne k}  \sum_{i_s=0}^{d_s} x_s^{i_s} 
$$
The product of power sums is greater than 1; divide both sides by it. Analogously for $M_k(f ; x_k)$.
\end{proof}

\begin{corollary} \label{cbounds}
Given $k\in\{1,\dott, n\}$, if $u\in \RR_+^n$ with $u_k\in]0,\mu_k]$,
where 
 $$
 \mu_k= 
 \left\{\begin{array}{cc}
 \text{min. pos. root of $M_k(f,x_k)$}  &  \text{if $M_k(f;0)<0$}\\  
 \text{min. pos. root of $m_k(f,x_k)$}  &  \text{if $m_k(f;0)>0$}\\
  0  & \text{otherwise}
 \end{array}\right. ,
 $$
then $f(u)\ne 0$. Consequently, all positive roots of $f$ lie in
$\RR_{>\mu_1}\times\cdots\times\RR_{>\mu_n}$. Also, for $u\in \RR_+^n$ with $u_k\in[\mathcal M_k,\infty]$,
 $$
\mathcal M_k= 
 \left\{\begin{array}{cc}
 \text{max. pos. root of $M_k(f,x_k)$}  &  \text{if $M_k(f;\infty)<0$}\\  
 \text{max. pos. root of $m_k(f,x_k)$}  &  \text{if $m_k(f;\infty)>0$}\\
  \infty  & \text{otherwise}
 \end{array}\right.,
 $$
it is $f(u)\ne 0$, i.e. all pos. roots are in $\RR_{<\mathcal M_1}\times\cdots\times\RR_{<\mathcal M_n}$.%

Combining both bounds we deduce that 
$[\mu_1,\mathcal M_1]\times\cdots\times [\mu_n,\mathcal M_n] $ 
is a bounding box for $f^{-1}(\{0\})\cap \RR_+^n$.
\end{corollary}
\begin{proof}
  The denominator in~(\ref{univbounds}) is always positive in
  $\RR_+^n$. Let $\uvec u\in\RR^n$ with $u_k\in[0,\mu_k]$. If
  $M_k(f,0)<0$ then also $M_k(f,u)<0$ and it follows $f(\uvec
  u)<0$. Similarly $m_k(f,0)>0 \Rightarrow m_k(f,u)>0 \Rightarrow
  f(\uvec u)<0$. The same arguments hold for $[\mathcal M_k,\infty]$,
  $M_k(f;\infty)=R(M_k(f;x_k))(0)$, $m_k(f;\infty)=R(m_k(f;x_k))(0)$,
  and $R(f)$, since lower bounds on the zeros of $R(f)$ yield upper
  bounds on the zeros of $f$.
\end{proof}

Thus lower and upper bounds on the $k-$th coordinates of the roots of
$(f_1,\dott,f_s)$ are given by 
\begin{equation}  \label{systembounds}
\max_{i=1,\dott,s}\{\mu_k(f_i) \} \ \ \ \ \text{ and }\ \ \  \
\min_{i=1,\dott,s}\{\mathcal M_k(f_i) \}
\end{equation}
respectively, i.e. the intersection of these bounding boxes.

We would like to remain in the ring of integers all along the process,
thus integer lower or upper bounds will be used.These can be the
floor or ceil of the above roots of univariate polynomials, or even
known bounds for them, e.g. Cauchy's bound.

If the minimum and maximum are taken with the ordering of coefficients
defined as $ c_i \prec c_j \iff c_i\binom{d}{j}\gamma^j\delta^{d-j} <
c_j\binom{d}{i}\gamma^i\delta^{d-i} $ then different $m_k(f,x_k),
M_k(f,x_k)$ polynomials are obtained. By Cor.~\ref{cor:bernsteincoefs}
their control polygon is the lower and upper hull respectively of the
projections of the tensor-Bernstein coefficients to the $k-$th
direction and are known to converge quadratically to simple roots when
preconditioning (described in the following paragraph)
is utilized~\cite[Cor.~5.3]{mp:smspe-05}.\\

\noindent{}{\bf Complexity analysis.} The analysis of the
subdivision step in Sect.~\ref{subsec:reduction} applies as well to the reduction step,
since reducing the domain means computing a new subdomain and ignoring the
remaining part.

If a lower bound $\uvec l$ is known, with $ \bitsize{l_k}=\sO(\sigma)$, then
the reduction step is 
performed in $\sOB(  n^2 d^n \tau + d^{n+1} n^3 \sigma )$. This is an
instance of Lem.~\ref{shiftComplexity}. 

The projections of Lem.~\ref{u_bounds} are computed using $\OO(d^n)$ comparisons. The computation of $\uvec l$ costs $\sOB(d^3\tau)$ in average, for solving these projections using univariate CF algorithm. Another option would to compute well known lower bounds on their roots, for instance Cauchy's bound in $\OO(d)$.
 
\noindent{}{\bf Illustration.} Consider a bi-quadratic $f_0\in\RR[x,y]$, namely, $\dg_{x_1} f_0= \dg_{x_2} f_0= 2$ with coefficients $c_{ij}$. Suppose that $f_0=H(f)$ for $I_0=I_H$. We compute
$$ m_1(f,x_1)= \sum_{i=0}^{2} \min_{j=0,\dott 2} c_{ij} \, x^i
\ \ \text{ and }\ \  
M_1(x)= \sum_{i=0}^{2} \max_{j=0,\dott 2}  c_{ij} \, x^i .$$
thus $ m(x) \leq \frac{f(x_1,x_2)}{1+x_2+x_2^{2}} \leq M(x) $. Fig.~\ref{fig:env} shows how these univariate quadratics bound the graph of $f$ in $I_0$.
\begin{figure}[h]
  \centering
  \includegraphics[trim= 0mm 20mm 0mm 30mm, clip, scale=.3]{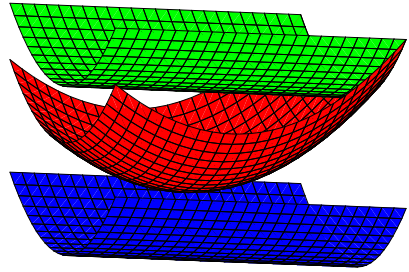}
  \caption{The enveloping polynomials $M_1(x),\, m_1(x)$ in domain $I_0$ for a bi-quadratic polynomial $f(x,y)$.}
  \label{fig:env}
\end{figure}

\subsection{Preconditioning}

To improve the reduction step, we use preconditioning. 
The aim of a preconditioner is to tune the system so that it can be
tackled more efficiently; in our case we aim at improving the bounds
of~Cor.~\ref{cbounds}.

A preconditioning matrix $P$ is an invertible $s\times s$ matrix that
transforms a system $(f_1,\dott,f_s)^t$ into the equivalent one
$P\cdot(f_1,\dott,f_s)^t$. 
This transformation does not alter the roots of the system, since the
computed equations generate the same ideal. The bounds obtained on the
resulting system can be used directly to reduce the domain of the
equations before preconditioning.
Preconditioning can be performed to a subset of the equations.

Since we use a reduction process using Cor.~\ref{cbounds} we
want to have among our equations $n$ of them whose zero locus
$f^{-1}(\{0\})$ is orthogonal to the $k-$th direction, for all
$k$.

Assuming a square system, we precondition $H(f_1),\dott,H(f_n)$ to
obtain a locally orthogonal to the axis system; an ideal
preconditioner would be the Jacobian of the system evaluated at a
common root; instead, we evaluate $J_{H(f)}$ at the image of the
center $\uvec u$ of the initial domain $I_H$,
$u_k=\frac{\alpha_k\delta_k+\beta_k\gamma_k}{2\gamma_k\delta_k
}$. Thus we must compute the inverse of the Jacobian matrix
$J_{H(\uvec f)}(x) = [\partial_{x_i}H(f_j)(x)]_{1\leq i,j,\leq n}$ evaluated
at $ \uvec u' := \mathcal H(\uvec u)= (\delta_1/\gamma_1,\dott,
\delta_n/\gamma_n)$.

\noindent{}{\bf Precondition step complexity.} Computing $J_{H(f)}(\uvec
u)\cdot (H(f_1),\dott,H(f_n))^t$ is done with cost $\sOB(n^2d^{n})$
and evaluating at $u'$ has cost $\sOB(n^2d^{n-1})$. We also need
$\sOB(n^2)$ for inversion and $\OO(n^2d^n)$ for multiplying
polynomials times scalar as well as summing polynomials. This
gives a precondition cost of order $\OO(n^2d^n)$.

\section{Exclusion -- Inclusion criteria} \label{criteria}

A subdivision scheme is able to work when two tests are
available: one that identifies empty domains (exclusion test) and one
that identifies domains with exactly one zero (inclusion test). If
these two tests are negative, a domain cannot be neither included nor
excluded so we need to apply further reduction/subdivision steps to it.
The certification is the following: if the result of the test is affirmative,
then this is undoubtedly true.

\noindent{\bf Exclusion test.}
The bounding functions defined in the previous section provide a fast
filter to exclude empty domains. 
Define $\min\{\}=\infty$ and $\max\{\}=0$.

\begin{corollary}
If for some $k\in\{1,\dott,n\}$ and for some $i\in\{1,\dott,s\}$ it is
$\mu_k(f_i)=\infty$ or $\mathcal M_k(f_i)=0$ then the system has no
solutions.  Also, if $\max_{i=1,\dott,s}\{\mu_k(f_i) \} >$ 
$\min_{i=1,\dott,s}$ $\{\mathcal M_k(f_i) \}$ then there can be no
solution to the system.
\end{corollary}
\begin{proof}
  For the former statement observe that $f_i$ has no real positive
  roots, thus the system has no roots. The latter statement means that
  the reduced domains of each $f_i,\, i=1,\dott,s$ do not intersect, thus
  there are no solutions.
\end{proof}

We can use interval arithmetic to identify additional empty domains; if
the sign of some initial $f_i$ is constant in $I_H= \mathcal
H(\RR_{>0}^n)$ then this domain is discarded. We can also simply inspect the
coefficients of each $H(f_i)$; if there are no sign changes then there
corresponding box contains no solution.

The accuracy of these criteria greatly affects the performance of the
algorithm. In particular, the sooner an empty domain is rejected the
less subdivisions will take place and the process will terminate
faster. We justify that the exclusion criteria will eventually succeed
on an empty domain by proving a generalization of Vincent's theorem to
the tensor multivariate case.

\begin{theorem} \label{vincentxyz} Let $f (\uvec x) = \sum_{\uvec
    i=\uvec 0}^{\uvec d} \, c_{\uvec i} \, \uvec x^{\uvec i}$ be a
  polynomial with real coefficients, such that it has no (complex)
  solution with $\Re (z_k) \ge 0$ for $k = 1, \dott, n$. Then all its
  coefficients $c_{i_1, \dott, i_n}$ are of the same sign.
\end{theorem}
\begin{proof}
  We prove the result by induction on $n$, the number of
  variables. For $n = 1$, this is the classical Vincent's
  theorem~{\cite{Vincent}}.
  
  Consider now a polynomial 
  \[f(x_1, x_2) = \sum_{0 \le i_1 \le d_1, 0 \le i_2 \le d_2} c_{i_1,
    i_2} \, x_1^{i_1} \, x_2^{i_2} \] in two variables with no
  (complex) solution such that $\Re (x_i) \ge 0$ for $i = 1, 2$. We
  prove the result for $n = 2$, by induction on the degree $d= d_1 +
  d_2$. The property is obvious for polynomials of degree $d = 0$. Let
  us assume it for polynomials of degree less than $d$.
  
  By hypothesis, for any $z_1 \in \mathbb {C}$ with $\Re (z_1) \ge 0$,
  the univariate polynomial $f (z_1, x_2)$ has no root with $\Re (x_2)
  \ge 0$.  According to Lucas theorem~\cite{Marden1966}, the complex
  roots of $\partial_{x_2} f (z_1, x_2)$ are in the convex hull of the
  complex roots of $f (z_1, x_2)$. Thus, there is no root of
  $\partial_{x_2} f (x_1, x_2)$ with $\Re (x_1) \ge 0$ and $\Re (x_2)
  \ge 0$. By induction hypothesis, the coefficients of
  $\partial_{x_2}f (x_1, x_2)$ are of the same sign. We decompose $P$
  as
  \[ f (x_1, x_2) = f (x_1, 0) + f_1 (x_1, x_2) \] where $f_1 (x_1,
  x_2) = \sum_{0 \le i_1 \le d_1, 1 \le i_2 \le d_2} c_{i_1, i_2} \,
  x_1^{i_1} \, x_2^{i_2}$ with $c_{i_1, i_2}$ of the same sign, say
  positive. By Vincent theorem in one variable, as $f (x_1, 0)$ has no
  root with $\Re (x_1) \ge 0$, the coefficients $c_{i_1, 0}$ of $f
  (x_1, 0)$ are also of the same sign. If this sign is different from
  the sign of $c_{i_1, i_2}$ for $i_2 \geqslant 1$ (ie. negative
  here), then $f (0, x_2)$ has one sign variation in its coefficients
  list. By Descartes rule, it has one real positive root, which
  contradicts the hypothesis on $f$. Thus, all the coefficients have
  the same sign.
  
  Assume that the property has been proved for polynomials in $n - 1$
  variables and let us consider a polynomial \ $f (\uvec x) = \sum_{i=0}^d
  \, c_{i} \, \uvec x^{i}$ in $n$ variables with \ no (complex) solution
  such that $\Re (x_k) \ge 0$ for $k = 1, \dott, n$. For any $z_1,
  \dott, z_{n - 1} \in \mathbb C$ with $\Re (z_k) \ge 0$, for $k = 1,
  \dott, n - 1$, \ the polynomial $f (z_1, \dott, z_{n - 1}, x_n)$
  and $\partial_{x_n} f (z_1, \dott, z_{n - 1}, x_n)$ has no root
  with $\Re (x_n) \ge 0$. By Lucas theorem and induction hypothesis on
  the degree, $\partial_{x_n} f (\uvec x)$ has coefficients of the same
  sign. We also have $f (x_1, \dott, x_{n - 1}, 0)$ with coefficients
  of the same sign, by induction hypothesis on the number of
  variables. If the two signs are different, then $f (0, \dott, 0,
  x_n)$ has one sign variation in its coefficients and thus one real
  positive root, say $\zeta_n$, which cannot be the case, since
  $(0,\dott,0,\zeta_n)$ would yield a real root of $f$. We deduce
  that all the coefficients of $f$ are of the same sign.
  
  This completes the induction proof of the theorem.
\end{proof}

This implies that empty regions will be eventually excluded by sign inspection.

\begin{corollary} \label{corxyz}
  Let $H(f) = \sum_{\uvec i=0}^{\uvec d} \, c_{\uvec i} \, \uvec x^{\uvec i}$
  be the representation of $f$ through $\mathcal H$ in a box $I_H=[\uvec u,
    \uvec v]$. If
  there is no toot $\uvec z\in\CC^n$ of $f$ such that  
  \[ \left|z_k - \frac{u_k + v_k}{2} \right| \le \frac{v_k-u_k}{2},
  \text{ for } k = 1, \dott, n, \] then all the coefficients $c_{i_1,
    \dott, i_n}$ are of the same sign.

That is, if $\ds \text{dist}_{\infty}({\mathcal Z}_{\CC^{n}}(f), m)>\varepsilon$,
where $m$ is the center of $I_H$, then $I_H$ is excluded by sign conditions. 
\end{corollary}
\begin{proof}
  The interval $[u_k,v_k]$ is transformed by $\mathcal H^{-1}$ into
  $[0,+\infty]$ and the disk $\left|z_k - \frac{u_k + v_k}{2} \right|
  \le \frac{v_k-u_k}{2}$ is transformed into the half complex plane
  $\Re(z_k)\geq 0$.  We deduce that $H(f)$ has no root with $\Re
  (z_k) \ge 0$, $k = 1, \dott, n$.  By Thm.~\ref{vincentxyz}, the
  coefficients of $H(f)$ are of the same sign.
\end{proof}

We deduce that if a domain is far enough from the zero locus of
some $f_i$ then it will be excluded, hence redundant empty domains
concentrate only in a neighborhood of $\uvec f=\uvec 0$.

\begin{definition}
The tubular neighborhood of size $\varepsilon$ of $f_i$ is the set 
$$
{\tau}_\varepsilon(f_i)=\{x\in\RR^n\ :\ \exists z\in\CC^{n},\, f_i(z)=0,\,\text{s.t. } \|z-x\|_\infty <\varepsilon\}  .  
$$
\end{definition}

We bound the number of boxes that are not excluded at each level of the subdivision tree.
\begin{lemma} Assume that for $\varepsilon_{0}>0$,
  $\cap_{i}{\tau}_{\varepsilon_{0}}(f_{i})\cap I_{0}$ is bounded.
Then the number of boxes of size $\varepsilon<\varepsilon_{0}$ kept by the algorithm
  is less than  $(1+{\sqrt{n}\over 2})^{n}\, c$, where $c>0$ is such that
  $\forall \varepsilon$ st. $\varepsilon_{0}>\varepsilon>0$,
$$ 
V(f, \varepsilon) := \mathrm{volume}\left(\cap_{i=1}^s \tau_\varepsilon (f_i) \cap
I_{0}) \right)
\leq c\, \epsilon^{n}.
$$
\end{lemma}
\begin{proof}
  Consider a subdivision of a domain $I_0$ into boxes of size
  $\varepsilon<\varepsilon_{0}$. We will bound the number $N$ of boxes in this
  subdivision that are not rejected by the algorithm. By
  Cor.~\ref{corxyz} if a box is not rejected, then we have for all
  $i=1,\dott,s$ $\text{dist}_{\infty}(\mathcal Z_{\CC^n}(f_i),
  m)<\varepsilon$, where $m$ is the center of the box.
  Thus all the points of this box are at distance $< \varepsilon (1+
  {\sqrt{n}\over 2})$ to $\mathcal Z_{\CC^n}(f_i)$
that is in $\cap_{i=1}^s \tau_{\varepsilon (1+
  {\sqrt{n}\over 2})}(f_i) \cap I_{0}$.

To bound $N$, it suffices to estimate the $n-$dimensional volume $V(f, 
  \varepsilon)$, since we have:
$$ 
N\varepsilon^{n} \leq \mathrm{volume}\left(\cap_{i=1}^s \tau_{ \varepsilon (1+
  {\sqrt{n}\over 2})}(f_i)\cap I_{0}\right)=V(f, 
  \varepsilon\,(1+ {\sqrt{n}\over 2}) ).
$$
 
When $\varepsilon$ tends to $0$, this volume becomes equivalent to a
constant times $\varepsilon^{n}$. For a square system with single
roots in $I_{0}$, it becomes equivalent to the sum for all real roots
$\zeta$ in $I_{0}$ of the volumes of parallelotopes in $n$ dimensions
of height $2\varepsilon$ and unitary edges proportional to the
gradients of the polynomials evaluated at the common root; It is thus
bounded by $\varepsilon^n 2^{n} \sum_{\zeta \in I_{0}} {|J_{f}(\zeta)|
  \over \prod_{i}||\nabla f_{i}(\zeta)||}$. We deduce that there
exists a constant $c \ge\, 2^{n}\, \sum_{\zeta \in I_{0}}
{|J_{f}(\zeta)| \over \prod_{i}||\nabla f_{i}(\zeta)||}$ such that
$V(f, \varepsilon)\le c\, \varepsilon^{n} < \infty$. For
overdetermined systems, the volume is bounded by a similar
expression. Since $V(f, \varepsilon) \varepsilon^{-n}$ has a limit
when $\varepsilon$ tends to $0$, we deduce the existence of the finite
constant $c$ and the bound of the lemma on the number of kept boxes of
size $\varepsilon$.
\end{proof}

\noindent{\bf Inclusion test.}
We present a test that discovers common solutions, in a box, or
equivalently in $\RR_{+}^n$, through homography. To simplify the
statements we assume that the system is square, i.e. $s=n$.

\begin{definition}
  The \emph{lower face} polynomial of $f$ w.r.t. direction $k$ is
  $\flower(f,k)=f|_{x_k=0}$. The \emph{upper face} polynomial of $f$
  w.r.t. $k$ is $\fupper(f,k)=f|_{x_k=\infty} := R_k(f)|_{x_k=0}$.
\end{definition}


\begin{lemma}[Miranda Theorem~\cite{vrahatis1989}]
  If for some permutation $\pi : \{1,\dott,n\}\to\{1,\dott,n\}$,
  $\sign( \flower( H(f_k),\pi(k) ) )$ and $\sign( \fupper( H(f_k),
  \pi(k) ) )$ are constant and opposite for all $k=1,\dott,n$, then
  the equations $(f_1,\dott,f_n)$ have at least one root in $I_H$.
\end{lemma}

The implementation of the Miranda test can be done efficiently if we
compute a $0-1$ matrix with $(i,j)-$th entry $1$ iff $\sign( \flower(
H(f_i),j))$ and $\sign(\fupper(H(f_i), j))$ are opposite. Then,
Miranda test is satisfied iff there is no zero row and no zero
column. To see this observe that the matrix is the sum of a
permutation matrix and a $0-1$ matrix iff this permutation satisfies
Miranda's test.

Combined with the following simple fact, we have a test that
identifies boxes with a single root.
\begin{lemma} 
  If $\det J_f(x)$ has constant sign in a box $I$, then there
  is at most one root of $f=(f_1,\dott,f_n)$ in $I$.
\end{lemma} 
\begin{proof}
  Suppose $u,v\in I$ are two distinct roots; by the mean value theorem
  there is a point $w$ on the line segment $\overline{uv}$, and thus
  in $I$, s.t. $J_f(w)\cdot (u-v)=f(u)-f(v)=\bm 0$ hence $\det
  J_f(w)=0$.
\end{proof}

\noindent{}{\bf Complexity of the inclusion criteria.}  Miranda test can be
decided with $\OO(n^2)$ evaluations on interval (cf.~\cite{Garloff00}) as well
as one evaluation of $J_{\uvec f}$, overall $\OO(n^2d^n)$ operations.  The
cost of the inclusion test is dominated by the cost of evaluating
$\OO(n)$ polynomials of size $\OO(d^n)$ on an interval,
i.e. $\OO(nd^n)$ operations suffice.

\begin{proposition} 
If the real roots of the square system in the initial domain $I_{0}$ are simple, then
Alg.~\ref{algo:subdivision} stops with boxes isolating the
real roots in $I_{0}$.
\end{proposition}
\begin{proof}
  If the real roots of $f=(f_{1}, \dott, f_{n})$ in $I_{0}$ are simple, in a small neighborhood of
  them the Jacobian of $f$ has a constant sign. By the inclusion
  test, any box included in this neighborhood will be output if and
  only if it contains a single root and has no real roots of the
  jacobian. Otherwise, it will be further subdivided or
  rejected. Suppose that the subdivision algorithm does not
  terminate. Then the size of the boxes kept at each step tends to
  zero. By Cor.~\ref{corxyz}, these boxes are in the intersection of
  the tubular neighborhoods $\left( \cap_{i=1}^s
    \text{tub}_\varepsilon (f_i) \right) \cap \RR^n$ for
  $\varepsilon>0$ the maximal size of the kept boxes. If $\varepsilon$
  is small enough, these boxes are in a neighborhood of a root in
  which the Jacobian has a constant size, hence the inclusion test
  will succeed.  By the exclusion criteria, a box domain is not
  subdivided indefinitely, but is eventually rejected when the
  coefficients become positive.  Thus the algorithm either outputs
  isolating boxes that contains a real root of the system or rejects
  empty boxes. This shows, by contradiction, the termination of the
  subdivision algorithm.
\end{proof}

\section{The complexity of  mCF} 
\label{sec:complexity}

In this section we compute a bound on the complexity of the algorithm
that exploits the continued fraction expansion of the real roots of
the system.  Hereafter, we call this algorithm MCF (Multivariate
Continued Fractions).  Since the analysis of the reduction steps of
Sec. \ref{subdivision-reduction} and the Exclusion-Inclusion test of
Sec.~\ref{criteria} would require much more developments, we simplify
the situation and analyze a variant of this algorithm. We assume that
two oracles are available. One that computes, exactly, the partial
quotients of the positive real roots of the system, and one that
counts exactly the number of real roots of the system inside a
hypercube in the open positive orthant, namely $\RR_+^n$ . In what
follows, we will assume the cost of the first oracle is bounded by
$\mathcal{C}_1$, while the cost of the second is bounded by
$\mathcal{C}_2$, and we derive the total complexity of the algorithm
with respect to these parameters.  In any case the number of reduction
or subdivision steps that we derive is a lower bound on the number of
steps that every variant of the algorithm will perform. The next
section presents some preliminaries on continued fractions, and then
we detail the complexity analysis.

\subsection{About continued fractions}

Our presentation follows closely \cite{emt-lncs-2006}.
For additional details we refer the reader to, e.g., \cite{Yap2000,BomPoo:contfrac:95,Poorten:intro}.
In general a {\em simple (regular) continued fraction} 
is a (possibly infinite)  expression of the form 
\begin{displaymath}
  c_0 +
  \cfrac{1}{
    c_1 + \cfrac{1}{
      c_2 + \dott  
    }
  } =
  [ c_0, c_1, c_2, \dott ],
\end{displaymath}
where the numbers $c_i$ are called {\em partial quotients}, 
$c_i \in \ZZ$ and $c_i \geq 1$ for $i > 0$.  
Notice that $c_0$ may have any sign, however, in our real root isolation
algorithm $c_0 \geq 0$, without loss of generality.
By considering the recurrent relations
\begin{displaymath}
  \begin{array}{cccc}
    P_{-1} = 1, & P_{0} = c_0, & P_{n+1} = c_{n+1}\, P_{n} + P_{n-1},\\
    Q_{-1} = 0, & Q_{0} = 1,   & Q_{n+1} = c_{n+1}\, Q_{n} + Q_{n-1},
  \end{array}
\end{displaymath}
it can be shown by induction that $R_n = \frac{P_n}{Q_n} = [c_0, c_1, \dott, c_n]$,
for $n=0,1,2,\dott$.

If $\gamma = [c_0, c_1, \dott]$ then 
$
\gamma = c_0 + \frac{1}{Q_0 Q_1} - \frac{1}{Q_1 Q_2} + \dott
= c_0 + \sum_{n=1}^{\infty}{ \frac{(-1)^{n-1}}{Q_{n-1}Q_n}}
$ 
and since this is a series of decreasing alternating terms it converges to some real
number $\gamma$.
A finite section $R_n = \frac{P_n}{Q_n} = [ c_0, c_1, \dott, c_n]$
is called the $n-$th {\em convergent} (or {\em approximant}) of $\gamma$
and the tails $\gamma_{n+1} = [c_{n+1}, c_{n+2}, \dott]$  are known as its
{\em complete quotients}. 
That is $\gamma = [ c_0, c_1, \dott, c_n,$ $ \gamma_{n+1}]$
for $n=0,1,2,\dott$.
There is an one to one correspondence between the real numbers and the continued
fractions, where evidently the finite continued fractions correspond
to rational numbers.

It is known that $Q_n \geq F_{n+1}$ 
and that $F_{n+1} < \phi^n < F_{n+2}$, 
where $F_n$ is the $n-$th Fibonacci number 
and $\phi = \frac{1+ \sqrt{5}}{2}$ is the golden ratio.
Continued fractions are the best rational approximation(for a given denominator size).
This is as follows:
\begin{equation}
  \label{eq:cf-approx}
  \frac{1}{Q_n(Q_{n+1} +Q_n)} 
  \leq \left| \gamma - \frac{P_n}{Q_n} \right| 
  \leq \frac{1}{Q_n Q_{n+1}} < \phi^{-2n+1}.
\end{equation}
Let $\gamma = [c_0, c_1, \dott]$ be the continued fraction expansion of a real
number. The Gauss-Kuzmin distribution
\cite{BomPoo:contfrac:95}
states that for almost all real numbers $\gamma$ (meaning that the set
of exceptions has Lebesgue measure zero) the probability for a
positive integer $\delta$ to appear as an element $c_i$ in the
continued fraction expansion of $\gamma$ is  
\begin{equation}
  Prob[ c_i = \delta ] \backsimeq
  \lg{\frac{(\delta+1)^2}{\delta(\delta+2)}}, \quad \text{for any fixed } i > 0.
\label{eq:Gauss-Kuzmin}
\end{equation}
The Gauss-Kuzmin law induces that we can not bound the mean value 
of the partial quotients 
or in other words that the expected value (arithmetic mean) of the partial
quotients is diverging, i.e. 
\begin{displaymath}
  E[c_i] = \sum_{\delta=1}^{\infty}{ \delta\, Prob[ c_i = \delta ]}= \infty,
  \text{ for } i >0.
\end{displaymath}

Surprisingly enough the geometric (and the harmonic) mean is not only
asymptotically bounded, but is bounded by a constant, for almost all $\gamma \in \RR$.
For the geometric mean this is the famous Khintchine's constant
\cite{Khintchine:97}, i.e.
\begin{displaymath}
  \lim_{n \rightarrow \infty}{\sqrt[n]{ \prod_{i=1}^{n}{c_i}}} = \mathcal{K} = 2.685452001...
\end{displaymath}
It is not known if $\mathcal{K}$ is a transcendental number.
The expected value of the bit size of the partial quotients is a constant
for almost all real numbers, 
when $n \rightarrow \infty$ or $n$ sufficiently big
 \cite{Khintchine:97}.
Notice that in (\ref{eq:Gauss-Kuzmin}), $i > 0$,
thus $\gamma \in \RR$ is uniformly distributed in $(0, 1)$. 
%
Let $\bitsize{c_i} \triangleq b_i$, then
\begin{equation}
  E[ b_i] = \OO( \lg{\mathcal{K}}) =\OO(1).
  \label{eq:exp_b}
\end{equation}

\subsection{Complexity results}

Let $\sigma$ be an upper bound on the bitsize of the partial quotient that
appear during the execution of the algorithm.

\begin{lemma}
  \label{lem:mcf-steps}
  The number of reduction and subdivision steps 
  that the algorithm performs is $\sO( n^2 \tau d^{2 n-1})$.
\end{lemma}
\begin{proof}
  Let $\zeta=(\zeta_{1},\dott, \zeta_{n} )$ be a real root of the system.
  It suffices to consider the number of steps needed to isolate
  the $i$ coordinate of $\zeta$.

  Recall, that we assume, working in the positive orthant, 
  we can compute exactly the next partial quotient in each coordinate;
  in other words a vector $l=(l_1,\dott, l_n)$, where each $l_i$, 
  $1 \leq i \leq n$, is the partial quotient of a coordinate of a positive 
  real\footnote{Actually the analysis holds even in the
  case where each $l_i$ is the partial quotient of the positive imaginary part of
  a coordinate of a solution of the system.} solution of the system.

  Let $k_i(\zeta)$ be the number of steps needed to isolate the
  $i$ coordinate of the real root $\zeta$. 
  The analysis is similar to the univariate case.
  The successive approximations of
  $\zeta_{i}$ by the lower bound $l_i$,
  yield  the $k_i(\zeta)$-th approximant,
  $\frac{P_{k_i(\zeta)}}{Q_{k_i(\zeta)}}$ of $\zeta_i$,
  which using (\ref{eq:cf-approx}) satisfies
  $$
  \left| \frac{P_{k_i(\zeta)}}{Q_{k_i(\zeta)}} - \zeta_{i}  \right|  
  \leq  \frac{1}{Q_{k_i(\zeta)}Q_{k_i(\zeta)+1}}
  <     \phi^{-2k_i(\zeta)+1}.  
  $$

  In order to isolate $\zeta_i$, it suffices to have
  $$
  \left| \frac{P_{k_i(\zeta)}}{Q_{k_i(\zeta)}} - \zeta_{i}  \right| 
  \leq  \Delta_{i}(\zeta),
  $$
  where $\Delta_{i}(\zeta)$ is the local separation bound of $\zeta_i$,
  that is the smallest distance between $\zeta_i$ and all the other
  $i$-coordinates of the positive real solutions of the system.

  Combining the last two equations, we deduce that
  to achieve the desired approximation, we should have 
  $\phi^{-2k_i(\zeta)+1} \leq \Delta_{i}(\zeta)$,
  or  
  $k_i(\zeta)\geq \frac{1}{2} - \frac{1}{2} \lg \Delta_{i}(\zeta)$.
  That is to isolate the $i$ coordinate it suffices to perform 
  $\OO(- \frac{1}{2} \lg \Delta_{i}(\zeta))$ steps.
  To compute the total number of steps, we need to sum over all positive real
  roots and multiply by $n$, which is the number of coordinates,
  that is 
  $$
  n \sum_{\zeta\in V} k_i(\zeta) 
  \leq n \frac{1}{2} R - n \frac{1}{2} \sum_{\zeta\in V} \lg \Delta_{i}(\zeta)  
  = n \frac{1}{2} R - n \frac{1}{2} \lg \prod_{\zeta\in V} \Delta_{i}(\zeta), 
  $$
  where $|V| = R$ is the number of positive real roots.

  To bound the logarithm of the product, we use $\texttt{DMM}_n$ \cite{emt-dmm-2009},
  i.e. aggregate separation bounds for multivariate, zero-di\-men\-sion\-al polynomial
  systems. It holds 
  \begin{displaymath}
    \begin{array}{lcl}
      \prod_{\zeta\in V}{\Delta_{i}(\zeta)} &\geq& 2^{-2n\tau d^{2n-1} - d^{2n}/2} \,  (nd^n)^{-n d^{2n}} \\
      -\log{ \prod_{\zeta\in V}{\Delta_{i}(\zeta)}} &\leq& 2n\tau d^{2n-1}  + 2 n  d^{n} \lg( n d^{2n}).
    \end{array}
  \end{displaymath}
  Taking into account that $R \leq d^n$ we conclude that the number of steps 
  is $\sO( n^2 \tau d^{2n-1})$.
\end{proof}

\begin{proposition}
  \label{prop:mcf-complexity}
  The total complexity of the algorithm is 
  $\sOB( 2^{n} n^7 d^{5n-1} \tau^2 \sigma + (\mathcal{C}_1 + \mathcal{C}_2) n \tau d^{n-1})$.  
\end{proposition}
\begin{proof}
  At each $h$-th step of algorithm, 
  if there are more than one roots of the corresponding system in the positive orthant
  (the cost of estimating this is $\mathcal{C}_2$,  
  we compute the corresponding partial quotients $l_h=(l_{h,1},\dott, l_{h,n})$,
  where $\bitsize{h_{h,i}} \leq \sigma_h$
  (the cost of estimating this is $\mathcal{C}_1$
  Then, for each polynomial of the system, $f$, we perform the shift operation
  $f( x_1 + l_1, \dots, x_n + l_n)$,
  and then we split to $2^n$ subdomains.
  Let us estimate the cost of the last two operations.

  A shift operation on a polynomial of degree $\leq d$, by a
  number of bitsize $\sigma$, increases the bitsize of the polynomial by an
  additive factor $n d \sigma$.
  At the $h$ step of the algorithm, the polynomials of the corresponding system
  are of bitsize $\OO( \tau + n d \sum_{i=1}^{h}{\sigma_h})$, and we need to
  perform a shift operation to all the variables, with number of bitsize
  $\sigma_{h+1}$.
  The cost of this operation is $\sOB( n d^n \tau + n^2 d^{n+1} \sum_{k=1}^{h+1}{\sigma_k})$,
  and since we have $n$ polynomials the costs becomes 
  $\sOB( n^2 d^n \tau + n^3 d^{n+1} \sum_{k=1}^{h+1}{\sigma_k})$,
  The resulting polynomial has bitsize $\OO( \tau + n d \sum_{k=1}^{h+1}{\sigma_k})$.
  
  To compute the cost of splitting the domain, we proceed as follows.
  The cost is bounded by the cost of performing $n 2^n$ operations $f(x_1 + 1, \dots, x_n + 1)$,
  which in turn is  
  $\sOB(n d^n \tau + n^2 d^{n+1} \sum_{k=1}^{h+1}{\sigma_k} + n^2 d^{n+1})$.
  So the total cost becomes
  $\sOB(2^{n} n^2 d^{n} \tau + 2^{n} n^3 d^{n+1} \sum_{k=1}^{h+1}{\sigma_k})$.

  It remains to bound $\sum_{k=1}^{h+1}{\sigma_k}$.
  If $\sigma$ is a bound on the bitsize of all the partial quotients that appear
  during and execution of the algorithm, then 
  $\sum_{k=1}^{h+1}{\sigma_k} = \OO( h \sigma)$.

  Moreover, $h \leq \#(T)  = \OO( n^2 \tau d^{2n-1})$ (lem.~\ref{lem:mcf-steps}),
  and so the cost of each step is
  $\sOB(2^{n} n^5 d^{3n} \tau \sigma)$.
  
  Finally, multiplying by the number of steps (lem.~\ref{lem:mcf-steps})
  we get a bound of $\sOB( 2^{n} n^7 d^{5n-1} \tau^2 \sigma)$.  

  To derive the total complexity we have to take into account that at each step
  we compute some partial quotients and  and we count the number of real root of
  the system in the positive orthant.
  Hence the total complexity of the algorithm is 
  $\sOB( 2^{n} n^7 d^{5n-1} \tau^2 \sigma + (\mathcal{C}_1 + \mathcal{C}_2) n \tau d^{n-1})$.  
\end{proof}

In the univariate case ($n=1$),
if we assume that (\ref{eq:exp_b}) holds for real algebraic numbers, 
then the cost of $\mathcal{C}_1$ and $\mathcal{C}_2$
is dominated by that of the other steps, that is the splitting operations, 
and the (average) complexity becomes $\sOB(d^3 \tau)$
and matches the one derived in \cite{et-tcs-2007}
(without scaling).

\subsection{Further improvements}
\label{sec:complexity-improvements}

We can reduce the number of steps that the algorithm performs, and thus improve
the total complexity bound of the algorithm, using the same trick as in
\cite{et-tcs-2007}.
The main idea is that the continued fraction expansion of a real root of a
polynomial does not depend on the initial computed interval that contains all
the roots.
Thus, we spread away the roots by scaling the variables of the polynomials of the
system by a carefully chosen value.

If we apply the map 
$(x_1, \dots, x_n) \mapsto (x_1/2^{\ell}, \dots, x_n/2^{\ell})$, to the initial
polynomials of the system, then the real roots are multiply by $2^{\ell}$, and
thus their distance increase.
The key observation is that the continued fraction expansion of the real roots
does not depend on their integer part.
Let $\zeta$ be the roots of the system, and $\gamma$, be the roots after the
scaling. It holds $\gamma =2^{\ell}\, \zeta$.
From \cite{emt-dmm-2009} it holds that 
\begin{displaymath}
  -\log{ \prod_{\zeta\in V}{\Delta_{i}(\zeta)}} 
  \leq  2n\tau d^{2n-1} \lg( n d^{2n})  + 2 n  d^{n} \lg( n d^{2n}),
\end{displaymath}
and thus
  \begin{align*}
    -\log{ \prod_{\zeta\in V}{\Delta_{i}(\gamma)}} 
    & =  
    -\log{ 2^{R \ell}  \prod_{\zeta\in V}{\Delta_{i}(\zeta)}}  \\
    & \leq 
    (2n\tau d^{2n-1} + 2 n  d^{n}) \lg( n d^{2n}) - R \,\ell.
  \end{align*}
If we choose $\ell = 2n d^{n-1}(d + \tau)\lg( n d^n)$
and assume that $R = d^{n}$ which is the worst case, then 
$-\log{ \prod_{\zeta\in V}{\Delta_{i}(\gamma)}}= 0$.
Thus, following the proof of Lem.~\ref{lem:mcf-steps},
the number of steps that the algorithm is $\OO( n d^n)$.

The bitsize of the scaled polynomials becomes
$\sO( n^2 d^{n+1} + n^2 d^{n} \tau)$.
The total complexity of algorithm is now
$$\sOB( 2^n n^5 d^{3n+1} \sigma + 2^{n} n^5 d^{3n} \tau 
+ n d^n( \mathcal{C}_1 + \mathcal{C}_2)),$$
where $\sigma$ the maximum bitsize of the partial quotient
appear during the execution of the algorithm.
If we assume that (\ref{eq:exp_b}) holds for real algebraic numbers, 
then $\sigma = \OO(1)$. Notice that in
this case, when $n=1$, the bound becomes $\sOB( d^3 \tau)$, which agrees with
the one proved in \cite{et-tcs-2007}.

The discussion above combined with Prop.~\ref{prop:mcf-complexity} lead us to:

\begin{theorem}
  \label{th:improved-mcf-complexity}
  The total complexity of the algorithm is 
  $\sOB( 2^n n^5 d^{3n+1} \sigma + 2^{n} n^5 d^{3n} \tau 
  + n d^n( \mathcal{C}_1 + \mathcal{C}_2))$.
\end{theorem}

\section{Implementation and Examples}\label{sec:impl}
\begin{figure}[t]
  \centering
  \includegraphics[scale=.27]{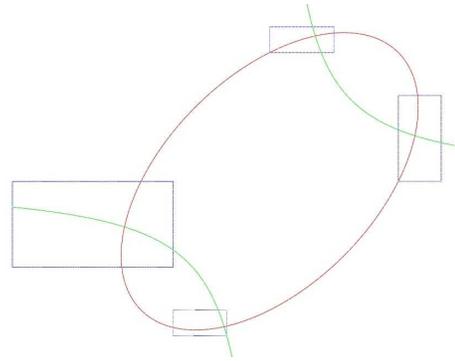}
  \caption{Isolating boxes of the real roots of $(\Sigma_1)$.}
  \label{fig:sys-1}
\end{figure}

\begin{figure*}[t]
  \centering
  \includegraphics[scale=.23]{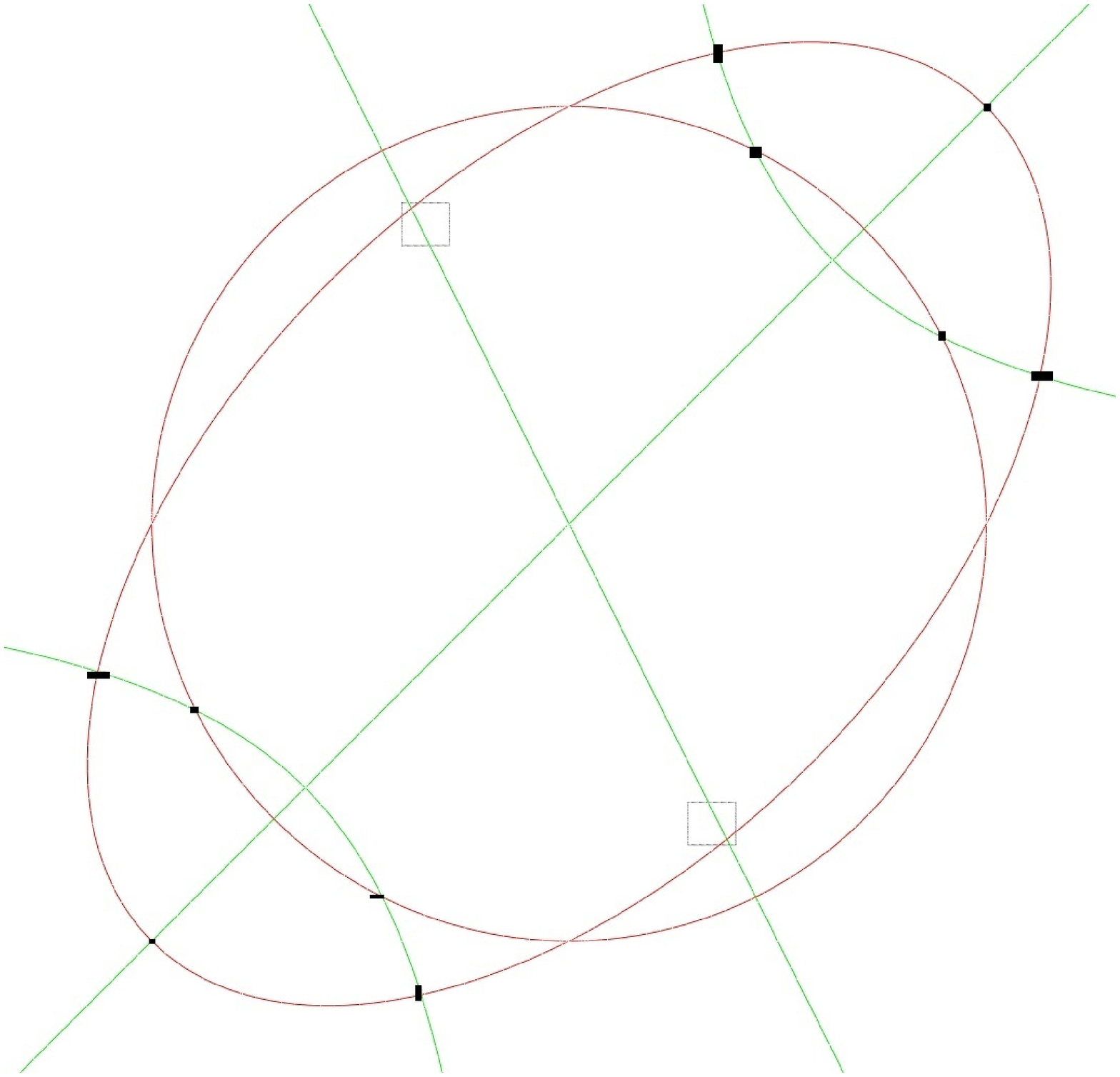}\hfill
  \includegraphics[scale=.23]{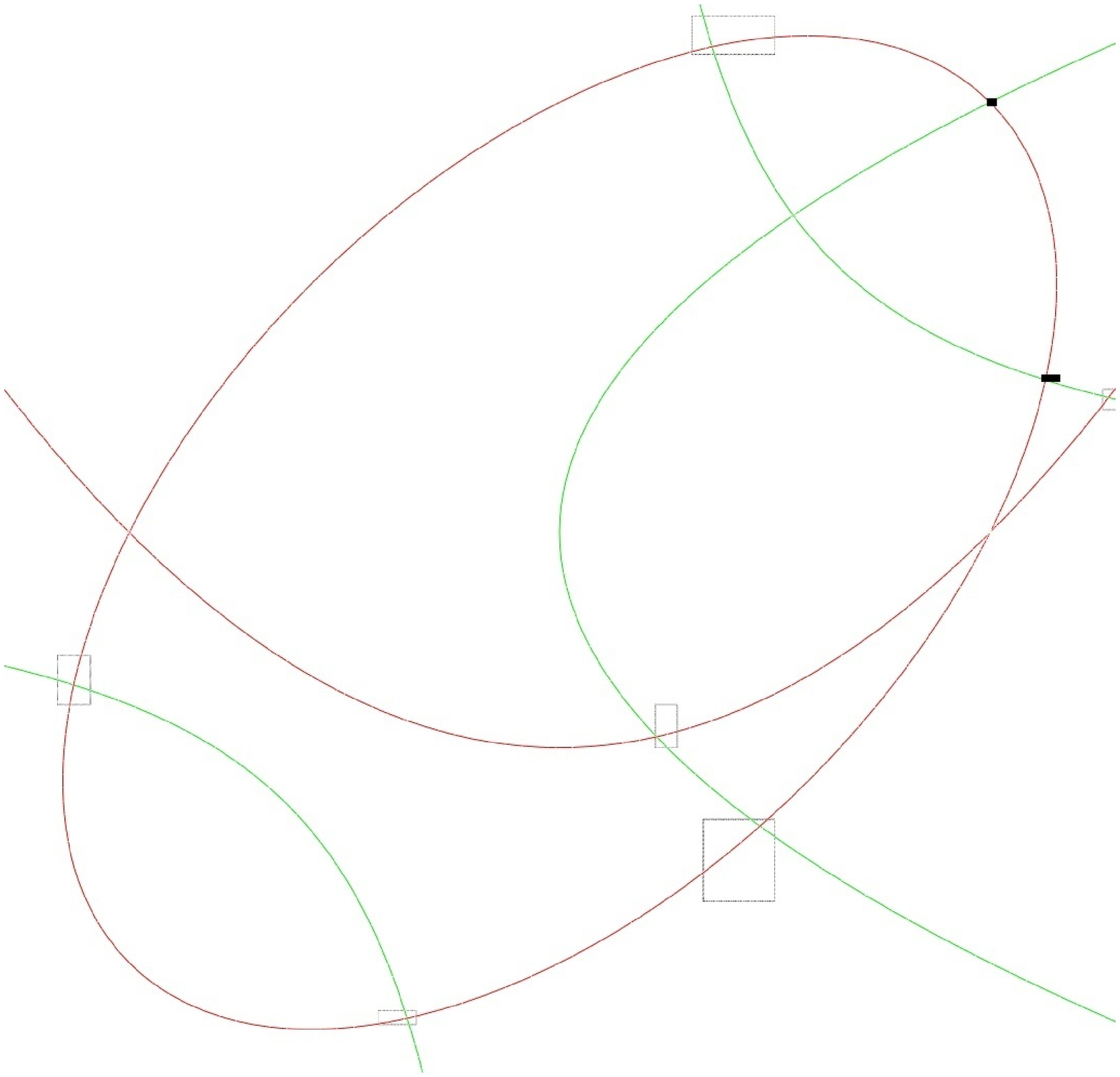}\hfill
  \includegraphics[scale=.23]{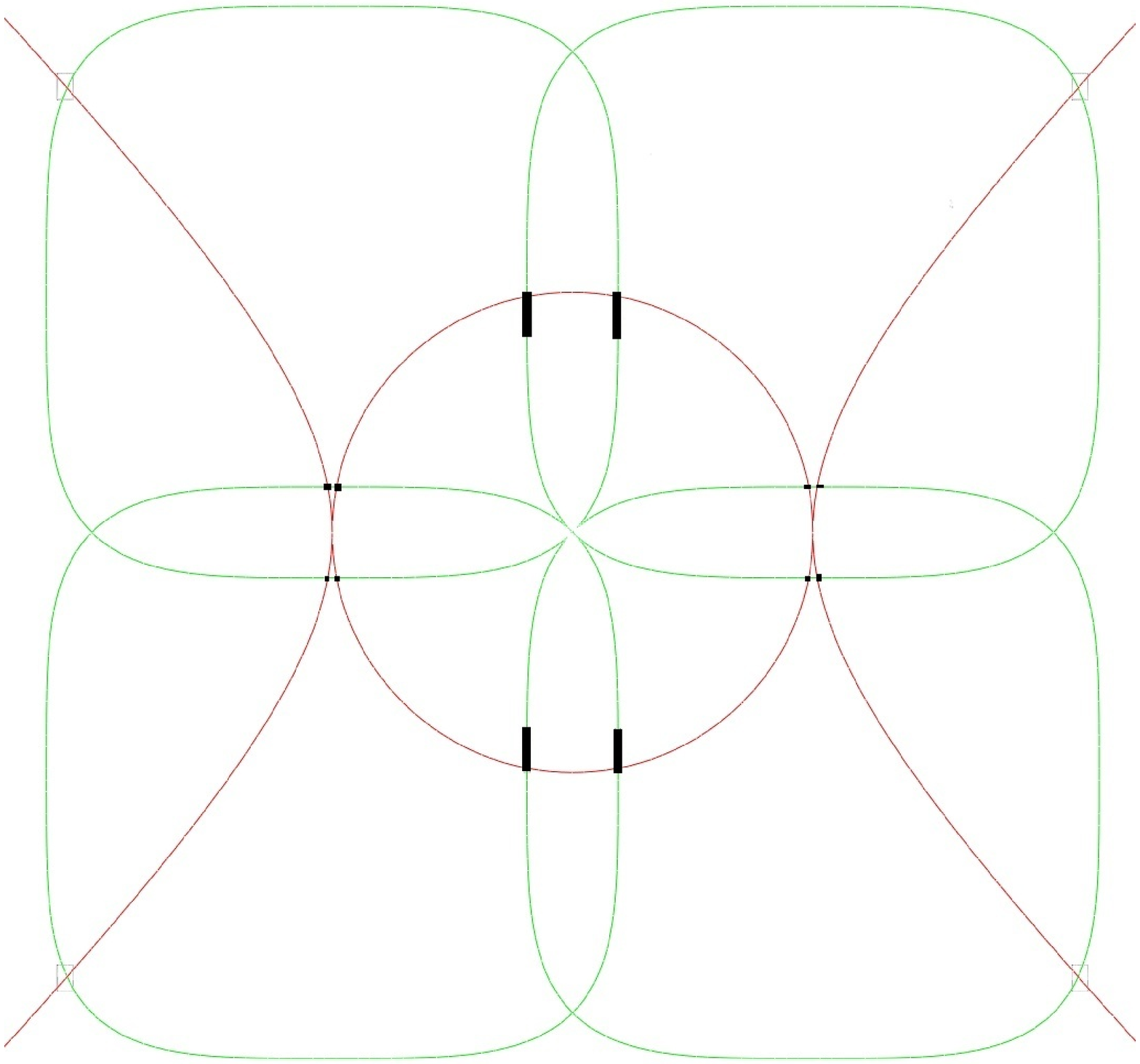}
  \caption{Isolating boxes of the real roots of the system 
    Left: $(\Sigma_2) $, Middle: $(\Sigma_3)$, Right: $(\Sigma_4)$.}
  \label{fig:sys-2}
  \label{fig:sys-3}
  \label{fig:sys-4}
\end{figure*}

We have implemented the algorithm in the C++ library \texttt{realroot} of  
\mathemagix\footnote{\texttt{http://www.mathemagix.org/}},
which is an open source effort that provides fundamental algebraic operations 
such as algebraic number manipulation tools, different types of univariate 
and multivariate polynomial root solvers, resultant and GCD computations, etc.

The polynomials are internally represented as a vector of coefficients
along with some additional data, such as a variable dictionary and the
degree of the polynomial in every variable. This allows us to map the
tensor of coefficients to the one-dimensional memory. The univariate
solver that is used is the continued fraction solver; this is
essentially the same algorithm with a different inclusion criterion,
the Descartes rule. The same data structures is used to store the
univariate polynomials, and the same shift/contraction routines. The
univariate solver outputs the roots in increasing order, as a result
of a breadth-first traverse of the subdivision tree. In fact, we only
compute an isolation box for the smallest positive root of univariate
polynomials and stop the solver as soon as the first root is
found. Our code is templated and is efficiently used with GMP
arithmetic, since long integers appear as the box size decreases.

The following four examples demonstrate the output of our implementation, 
which we visualize using \axel\footnote{\texttt{http://axel.inria.fr}}.

First, we consider the system $f_1 = f_2 = 0$ ($\Sigma_1$), 
where $f_1= x^2+y^2-xy-1$, and $f_2= 10xy-4$.
We are looking for the real solutions in the domain 
$I=[-2,3]\times[-2,2]$, which is mapped to $\RR^2_+$, by an initial
transformation. The isolating boxes of the real roots can be seen in
Fig.~\ref{fig:sys-1}. 

In systems $(\Sigma_2),\, (\Sigma_3)$, We multiply $f_1$ and $f_2$ by
quadratic components, hence we obtain
$$
(\Sigma_2) \left\{
\begin{array}{l}
f_1= x^4+2x^2y^2-2x^2+y^4-2y^2-x^3y-xy^3+xy+1\\
f_2= 20x^3y-10x^2y^2-10xy^3-8x^2+4xy+4y^2
\end{array}
\right.
$$
and 
$$
(\Sigma_3)
\left\{
\begin{array}{l}
f_1= 10x^2y-10xy^3-4x+4y^2  \\
f_2= x^4-2x^2y-2x^2+y^2x^2-2y^3-y^2-x^3y+\\ 
\hspace{.7cm} 
+2xy^2+xy+2y+1
\end{array}
\right.
$$
The isolating boxes of this system could be seen in
Fig.~\ref{fig:sys-2}.  Notice, that
size of the isolation boxes that are returned in this case is
considerably smaller. 

Consider the system $(\Sigma_4)$, which 
consists of $f_1= x^4-2x^2-y^4+1$ and $f_2$, which is a polynomial of bidegree $(8,8)$.
%
The output of the algorithm, that is the isolating boxes of the real roots can be seen in 
Fig.~\ref{fig:sys-4}.
One important observation is the fact the isolating boxes {\em are not} squares,
which verifies the adaptive nature of the proposed algorithm.

We provide execution details on
these experiments in
Table~\ref{tab:exec}. Several
optimizations can be applied to our
code, but the results already
indicate that our approach competes
well with the Bernstein case.

\begin{table} \begin{center}\begin{tabular}{|c|c|c|c|c|c|} \hline
System     & Domain   &Iters.& Subdivs. & Sols. & Excluded \\ \hline \hline
$\Sigma_1$ &$[0,10]^2$&53    &26        & 4     & 25   \\ \hline
$\Sigma_2$ &$[-2,3]^2$&263   &131       &12     & 126  \\ \hline
$\Sigma_3$ &$[-2,3]^2$&335   &167       &8      &160   \\ \hline
$\Sigma_4$ &$[-3,3]^2$& 1097 & 548      & 16    &533   \\ \hline
\end{tabular}\end{center}
\caption{Execution data for $\Sigma_1$, $\Sigma_2$, $\Sigma_3$, $\Sigma_4$.}
\label{tab:exec}
\end{table}

\vspace{.1cm}

\noindent{\bf Acknowledgements}\\
The first and second author were supported by Marie-Curie Initial
Training Network SAGA, [FP7/2007-2013], grant
[PITN-GA-2008-214584]. The third author was supported by contract
[ANR-06-BLAN-0074] ``Decotes''.

\bibliographystyle{abbrv}

\end{document}